\documentclass[11pt]{article}
\usepackage{amssymb,amsmath,amsthm}

\usepackage[margin=1in]{geometry}
\usepackage[pagebackref]{hyperref}

\hypersetup{
bookmarks=true,
colorlinks=true,
linkcolor=blue,
urlcolor=blue,
citecolor=blue,
pdftex,
linktocpage=true, 
hyperindex=true,
}

\begin{document}


\newtheorem{theorem}{Theorem}[section]
\newtheorem{conjecture}[theorem]{Conjecture}
\newtheorem{lemma}[theorem]{Lemma}
\newtheorem{corollary}[theorem]{Corollary}
\newtheorem{proposition}[theorem]{Proposition}
\newtheorem{definition}[theorem]{Definition}
\newtheorem{claim}{Claim}
\newtheorem{fact}{Fact}
\newtheorem{obs}{Observation}
\renewcommand{\leq}{\leqslant}
\renewcommand{\le}{\leqslant}
\renewcommand{\geq}{\geqslant}
\renewcommand{\ge}{\geqslant}

\newcommand{\gap}{\vspace{-1ex}}
\newcommand{\gapp}{\vspace{-2ex}}

\newcommand{\zero}{\textsf{0}}
\newcommand{\one}{\textsf{1}}
\setcounter{page}{0}
\title{Explicit two-deletion codes with redundancy matching the existential bound}
\author{Venkatesan Guruswami\thanks{Computer Science Department, Carnegie Mellon University, Pittsburgh, USA. Email: {\tt venkatg@cs.cmu.edu}. Research supported in part by NSF grant CCF-1814603.} \and Johan H\aa stad\thanks{Department of Mathematics, School of Engineering Sciences, KTH Royal Institute of Technology, Stockholm, Sweden. Email: {\tt johanh@kth.se}. Research supported by a grant from the Knut and Alice Wallenberg Foundation.}}
\date{}

\maketitle
\thispagestyle{empty}

\begin{abstract}
	We give an explicit construction of length-$n$ binary codes capable of correcting the deletion of two bits that have size $2^n/n^{4+o(1)}$. This matches up to lower order terms the existential result, based on an inefficient greedy choice of codewords,  that guarantees such codes of size $\Omega(2^n/n^4)$. Our construction is based on augmenting the classic Varshamov-Tenengolts construction of single deletion codes with additional check equations. We also give an explicit construction of binary codes of size $\Omega(2^n/n^{3+o(1)})$ that can be \emph{list decoded} from two deletions using lists of size two. Previously, even the existence of such codes was not clear.
\end{abstract}
\tableofcontents
\newpage
\section{Introduction}

We study deletion-correcting codes over the binary alphabet. Specifically, we are interested in codes $C \subset \{0,1\}^n$ such that if a codeword $x \in C$
is corrupted by deleting up to $k$ bits to obtain a subsequence
$y \in \{ 0,1\}^{n-k}$, then one can reconstruct $x$ from
 $y$.  Crucially, the location of the deleted bits are unknown. The parameter $k$ bounds the maximum number of deletions the code is designed to correct. The $k$-deletion correcting property is equivalent to the property that the length of the longest common subsequence between any two distinct codewords is less than $n-k$.
 The goal is to find codes of as large a size as possible that can correct up to $k$ deletions.

 For the case of fixed $k$ and growing $n$, which is the regime of interest in this paper, the size of the optimal $k$-deletion correcting code, say $D(n,k)$, satisfies
 \[ \Omega_k\Bigl(\frac{2^n}{n^{2k}} \Bigr) \le D(n,k) \le O_k\Bigl(\frac{2^n}{n^k}\Bigr) \ , \]
 where $O_k(\cdot)$ and $\Omega_k(\cdot)$ suppress factors that depend only on $k$.
The lower bound follows by a simple (but inefficient) greedy construction of picking codewords no two of which share a common subsequence of length $n-k$. The upper bound follows from a packing argument since the length $n-k$ subsequences of various codewords have to be distinct, and a typical  string has $\Omega_k(n^k)$ such subsequences. Defining the redundancy of a code $C$ to be $n - \log_2 |C|$ (since $n$ bits are transmitted to communicate one of $|C|$ possible messages), the optimal redundancy of $k$-deletion codes is between $k \log_2 n$ and $2k \log_2 n$ (ignoring additive constants depending on $k$).


For the single deletion case, the Varshamov-Tenengolts (VT) construction~\cite{VarshamovTenengolts65} is an explicit code of asymptotically optimal size $\Theta(2^n/n)$ as shown by Levenshtein~\cite{Levenshtein66} over 50 years ago. This codes consists of all strings $x \in \{0,1\}^n$ for which $f_1(x) := \sum_{i=1}^n i x_i \equiv 0 \pmod{n+1}$.  The next simplest case of two deletions, however, already turns out to be much more challenging, and attempts to recover from two deletions by augmenting the
 VT construction with various natural additional check equations have not met with success.

For $k \ge 2$, closing the gap between the lower and upper bounds on redundancy, and finding explicit constructions that come close to the existential bound (i.e., with redundancy $\approx 2k\log_2 n$), are two central challenges that still remain open. This work considers the latter question for the case $k=2$. By an explicit construction, we mean a code of length $n$ that has a deterministic $\text{poly}(n)$ time encoding algorithm. Until recently, the constructions of $k$-deletion codes even for $k=2$ had redundancy $\Omega(n)$~\cite{HF02,PGFC12}. A construction with redundancy about $O(\sqrt{n})$ is implicit in the work~\cite{GW-ieeetit} which considered high rate codes for correcting a small \emph{fraction} of deletions. 

Explicit constructions of size $2^n/n^{O(1)}$, i.e., $O(\log n)$ redundancy, were only constructed recently. Specifically, $k$-deletion codes of redundancy $O(k^2 \log k \log n)$ were constructed in \cite{BGZ-soda}. Following this, a sequence of works starting with 
\cite{Belazzougui15} studied $k$-deletion codes in the framework of deterministic document exchange protocols, leading to codes with redundancy $O(k \log^2 (n/k))$~\cite{CJLW18,Haeupler19} and even $O(k \log n)$ for small $k$~\cite{CJLW18}. (In the document exchange problem, Alice holds $x \in \{0,1\}^n$ and Bob holds a subsequence $y$ of $x$ with $k$ deletions, and Alice must send a \emph{short} ``sketch" $s(x)$ to Bob that will enable Bob to recover $x$. When $s(x)$ is a deterministic function of $x$, such a protocol is closely connected to $k$-deletion codes with redundancy roughly equal to the length of the sketch; see Section~\ref{subsec:reduce-to-sketch} for more on this connection.)

However, these constructions use hashing based recursive approaches and other ideas, and even for $k=2$ will have redundancy $C \log n$ for a rather large constant $C$.
For the case of two deletions specifically, two recent works constructed codes with redundancy $\approx 8 \log_2 n$~\cite{GS19} and $\approx 7 \log_2 n$~\cite{SB19}. The construction in \cite{GS19} followed the rough approach in \cite{BGZ-soda} based on hashing the pattern of occurrence of certain substrings in the string, and also considered several cases based on the identity and location within runs of the two bits deleted. The construction in \cite{SB19} is more explicit and can be viewed as a higher-dimensional version of the VT construction with certain modular check sums constraining indicator vectors of the string.

\medskip \noindent \textbf{Our results.}
In this work, we present an explicit construction of $2$-deletion codes in the mold of VT codes with redundancy close to the existential $4 \log_2 n$ bound. In addition to the position-weighted VT-sketch $f_1(x)$ of the string $x \in \{0,1\}^n$ to be recovered, we also include a quadratically-weighted VT-like sketch as well as a sketch based on the run-number sequence of the string.  These are the functions $f_1(x), f_2(x)$ and $f_1^r(x)$ defined in Equations \eqref{eq:f1}-\eqref{eq:f1r}. The goal is to recover $x$ from any subsequence $y$ formed by deleting two bits and the knowledge of these functions. 
If the two deleted bits are both \zero's or \one's, $f_1(x)$ and $f_2(x)$ together with $y$ suffice to reconstruct $x$. When one \zero\ and one \one\ are deleted, we bring the run-number sequence into the picture. The two deletions alter the number of runs by $0,2$ or $4$. When the run count changes by $0$ or $4$, the values $f_1(x)$ and $f_1^r(x)$ together with $y$ suffice to reconstruct $x$. The remaining case when one \zero\ and one \one\ are deleted and the number of runs decreases by $2$ turns out to be a lot harder. In this situation, we prove that $f_1(x), f_2(x)$ and $f_1^r(x)$ together localize the deletions to a small $O(\log n)$ long stretch of $x$, assuming that $x$ has a certain regularity property (namely that $x$ contains substrings $00$ and $11$ often enough). 
To finish the recovery, we employ a less efficient sketch of $O(\log \log n)$ bits that enables recovery of two deletions in $O(\log n)$-length strings. To satisfy the regularity assumption, we encode messages into regular strings with negligible rate loss.

Our final construction combining these ingredients gives $2$-deletion correcting codes with a redundancy matching the best known existential bound of $4 \log_2 n$ up to lower order terms.
\begin{theorem}
\label{thm:unique-decoding-intro}
There is an explicit (efficiently encodable) $2$-deletion correcting code $C \subseteq \{0,1\}^n$ with redundancy $4 \log_2 n + O(\log \log n)$. 
\end{theorem}

As a warm-up to the above construction, we also present a new code to tackle the single deletion case based on the run length sequence (specifically the sketch $f_1^r(x)$ defined in \eqref{eq:f1r}). While slightly more redundant than the VT code, by including a quadratic version  of this run-based sketch (namely $f_2^r(x)$ defined in \eqref{eq:f12r}), we also give a $2$-deletion code with redundancy smaller than the existential $4 \log_2 n$ bound at the expense of pinning down the codeword to one of two possibilities.

\begin{theorem}
\label{thm:list-decoding-intro}
There is an explicit (efficiently encodable) code $C \subseteq \{0,1\}^n$ with redundancy $3 \log_2 n + O(\log \log n)$ that can be list decoded from two deletions with a list of size $2$.
\end{theorem}

For the decoding, we can of course recover the two missing bits in quadratic time by trying all possible placements, only one of which (or at most two of which, in the case of Theorem~\ref{thm:list-decoding-intro}) will match the sketches. 
However, we can in fact perform the decoding in linear time. Once we find a single placement of the bits that gets the correct value for the VT sketch $f_1(x)$ correct, the algorithm consists of sweeping each of the bits either left or right across the string just once, and the updates to the sketches can be maintained online in $O(1)$ time per move (on the RAM model where operations on $O(\log n)$ bit integers take constant time). For simplicity, we do not elaborate on the linear complexity decoding further but an interested reader can verify this based on the details of our (algorithmic) proof of the $2$-deletion correction property.


 It is well known that a code capable of correcting $k$ deletions is capable of correcting any combination of up to a total of $k$ insertions and deletions~\cite{Levenshtein66}. However, this doesn't necessarily preserve the efficiency of the decoding algorithm. We have not explored decoding strategies from two insertions for our codes (of course the naive quadratic time approach still applies).
The case of insertion/deletion combination is more subtle for list decoding (see for instance the recent work~\cite{GHS-stoc20}), and we did not investigate how our list-decodable codes behaves under insertions.

Our work raises the intriguing possibility that it might be possible to achieve a redundancy smaller than $4 \log_2 n$ for $2$-deletion codes,  which would be quite exciting (and perhaps surprising) progress on this classical question.  Another natural question is whether our methods can extended to the case of more deletions. This appears rather difficult already for three deletions due to the many more combinations in which bits can be inserted. 

\medskip\noindent\textbf{Outline.} In Section~\ref{sec:prelims}, we reduce the task of constructing deletion-correction codes to finding good short sketch functions that together with any subsequence enable recovery of the original string, and also describe the sketch functions we will use in our constructions. As a warm-up, in Section~\ref{sec:single-deletion} we present our run-sequence based construction of a single-deletion code which also lets us establish the framework of moving the to-be-inserted bit(s) that we use to analyze all our constructions. We then present our construction of $2$-deletion codes for list decoding with size two lists in Section~\ref{sec:list-decode}. Finally, we give the $2$-deletion code establishing our main result Theorem~\ref{thm:unique-decoding-intro} in Section~\ref{sec:unique-decode}.

\gap
\section{Preliminaries}
\label{sec:prelims}
In our basic setup, we have an unknown string $x \in \{ 0,1\}^n$ which
is corrupted by deleting up to $k$ bits to obtain a subsequence
$y \in \{ 0,1\}^{n-k}$. The goal is to reconstruct $x$ from $y$. The location(s) of the deleted bits
are unknown. Our focus in this work is on the case $k=2$, though we will consider the single deletion case a warm-up to our construction for tackling two deletions.

\gap
\subsection{Reduction to recovery from known sketches}
\label{subsec:reduce-to-sketch}
If an arbitrary $x$ is allowed this is clearly an impossible task
and we are interested in \emph{codes} $C$, which are carefully constructed subsets of $\{0,1\}^n$, such that under the guarantee
that $x \in C$, the reconstruction is always possible. The goal is to maximize the size of $C$.  
There are many possible ways to construct a set $C$
but in this paper we are interested in the case
when there are one or more integer valued functions $(f_i)_{i=1}^t$ such
that knowing the value of $f_i(x)$ for $1 \leq i \leq t$ (which we can think of as \emph{sketches} or deterministic hashes of $x$) and the subsequence $y$,
it is possible to reconstruct $x$.  If there are only $T$ possible
values of $(f_i(x))_{i=1}^t$ then this implies the existence of a code $C$ of 
size at least $2^n/T$ for which reconstruction is possible.
Indeed, one can take $C$ to be the pre-image of the most common value for
these outputs.

However, an explicit description of the strings attaining this most frequent value is necessary in order to construct and efficiently encode into the code $C$. Even for modestly complex functions $f_i(\cdot)$, this can be difficult. Instead, below we give an alternate (standard) reduction of the code construction problem to recovering the string $x$ from its (known) sketches and the subsequence $y$. The idea is simply to encode the sketches, which are much shorter, by a known but less efficient $k$-deletion correcting code. We can then encode a message $x$ by appending these encoded sketches to $x$ itself. The formal proof is omitted as it implicitly appears in several previous works, including \cite{BGZ-soda}.
\begin{lemma}
	\label{lem:reduce-to-sketch}
Fix an integer $k\ge 1$. 
Let $s : \{0,1\}^n \to \{0,1\}^{\lceil c \log n \rceil+O(1)}$ be an efficiently computable function.
Suppose that $x \in \{0,1\}^n$ can be recovered from $s(x)$ and $y$ for any subsequence $y \in\{0,1\}^{n-k}$ of $x$. Then there is an efficiently encodable map $E$ mapping strings of length $n$ to strings of length  $N \le n +  c\log n + O_k(\log \log n)$ such that the image of $E$ is a $k$-deletion correcting code. In other words, we have a length code $C \subset \{0,1\}^N$ of size $2^{N}/N^{c+o(1)}$ that can correct $k$ deletions and into which we can efficiently encode.  
\end{lemma}

Given the above lemma, we turn to the definition of suitable sketches of total length $c \log n$ for as small $c$ as possible that help with recovery from (two) deletions. 
For our construction, the recovery will be guaranteed only certain ``regular" $x$ which constitute most of the strings but not all of them. So in order to obtain deletion codes out of our construction we will also need to encode into regular strings, which we will handle separately on top of Lemma~\ref{lem:reduce-to-sketch}.

\gap
\subsection{Position and run based sketches}
For a binary string $x \in \{ 0,1\}^n$, with $x_1$ as the
first bit,  we define the
{\em run string} $r$ associated with it as follows.
To make the arguments slightly more uniform avoiding special
cases at the beginning and end of $x$, we artificially insert a zero 
before $x$ and add a one at the end of $x$.  Thus we have $x_0=0$, $x_{n+1}=1$,
$r_0=0$ and set $r_{i+1}=r_i$ if $x_{i+1}=x_i$ and $r_{i+1}=r_i+1$
otherwise for $0 \leq i\leq n$. The quantity $r_i$ is referred to as the \emph{rank} (or run number) of the $i$'th bit of $x$. 

Note that with the inclusion of $x_0$ and $x_{n+1}$ at either end of a subsequence $y$ of $x$, the insertion of a bit into $y$ creates either no run or exactly two runs, even if the insertion happens at either end (just to the right of $x_0$ or just to the left of $x_{n+1}$).

To see an example if $x=001000111010$ then we first add the 
extra bits obtaining $(0)001000111010(1)$ producing the run
string $(0)001222333456(7)$.
Clearly there is a one-to-one
correspondence between run strings and binary strings.

Given a string $x$ we define some ``sketch" functions
of interest.
\begin{align}
  f_1(x)& = \sum_{i=1}^{n} i \cdot x_i   \label{eq:f1}  \\
f_2(x) &=  \sum_{i=1}^{n} {i \choose 2}  \cdot x_i  \label{eq:f2} \\
f^r_1(x)& =  \sum_{i=1}^{n+1} r_i  \label{eq:f1r} \\
f^r_2(x) & = \sum_{i=1}^{n+1} {r_i \choose 2}   \label{eq:f12r}
\end{align}
Note that we include $r_{n+1}$ in the sums but not $x_{n+1}$.
This is of no great consequence but simply convenient.

It is easy to see that $ 0 \leq f_1(x) \leq n(n+1)/2$ and thus
it seems like we would need $\Omega(n^2)$ values for $f_1(x)$, but
a moment's reflection indicates that we can do significantly better.
Suppose we are in the one-deletion case and we are given $y$
and we try to reconstruct $x$.
It is easy to see that $f_1(y) \leq f_1(x) \leq f_1(y)+n$.
Thus it is sufficient to give the value 
of $f_1(x)$ modulo $n+1$ and then use $y$ to reconstruct
$f_1(x)$ over the integers.  In the case of two deletions
it is sufficient to specify the same number modulo $2n+1$.

We will not be particularly careful with constant factors in the \emph{size} of the code.
 Let us simply note that it is sufficient to specify $f_2(x)$ 
and $f_2^r(x)$ modulo a number that is $O(n^2)$.  
For $f_1^r (n)$ the corresponding number is $O(n)$.  This information,
together with $y$, makes it possible to reconstruct these numbers over
the integers.

\medskip\noindent {\bf Constant-sized sketches.}
Finally it is several times convenient to know the total number of runs in $x$
as well as the number $\sum_{i=1}^n x_i$, the number of
ones in $x$.  It is sufficient to specify these quantities
modulo a number that is $O(1)$.  These two
numbers are not needed in all the reconstruction
algorithms but we assume they are available whenever needed.

\gap
\section{The single deletion case and moving bits}
\label{sec:single-deletion}
\gap
We begin by developing our ideas in the simpler context of recovery from one deleted bit.
When given $y$ it
is many times convenient to insert
the missing bit(s) in some position(s) in $y$ possibly giving 
the correct value for one of the functions and see what possible
changes can be done maintaining the already established
equality.  In particular if we can make the output of another function
monotone under these changes it follows that we have a unique placement of the missing bits.

As a simple example let us analyze the (well-known) single deletion case.
We think of the string $x$ written from left to right
starting with $x_1$.  If $x$ is formed by inserting
a 0 in position $i$ then $x_j=y_j$ for $1 \leq j < i$,
$x_i=0$ while $x_j=y_{j-1}$ for $i< j \leq n$.
Moving  the inserted bit one step to the left means
forming a new string $x'$ with $x'_j=x_j$ for $j \not\in \{ i-1,i\}$
while $x'_{i-1}=0$ and $x'_{i}=x_{i-1}$.  Moving the
bit to the right is defined analogously.  We find it
easier to use $x$ to denote a dynamic string and
hence we mostly abstain from using $x'$. The following is the basis of the single deletion correcting property of the VT code which is defined as the set of strings $x$ with $f_1(x) \equiv 0 \pmod {(n+1})$.

\begin{theorem}\label{thm:1sum}
In the case of one deletion,
the value of $f_1(x)$ modulo $n+1$ jointly with $y$ determines
$x$ uniquely.
\end{theorem}
\begin{proof}
Let us first insert a bit with the value \zero\ at the very
end of $y$, i.e. setting $x_i=y_i$ for $1 \leq i \leq n-1$
and $x_n=0$.  Clearly in this case we have
$f_1(x)=f_1(y)$. 
Now keeping this bit as a \zero\ and moving it left, the value of $f_1(x)$
increases by one each time the moving \zero\ passes
a \one.  When the moving \zero\ passes another \zero, the value
of $f_1(x)$ does not change, but this is natural as $x$ does
not change, only the information which of its bits come from $y$ and where
the inserted zero is placed changes.

Once the moving \zero\ has moved all the way to become $x_1$ we change
its value to \one.  Also this increases the value of $f_1(x)$ by
one.  Finally moving this \one\ to the right, each time the
\one\ passes a \zero, $f_1(x)$ increases by one and finally when
the \one\ is inserted as $x_n$ the value of $f_1(x)$ is
$f_1(y)+n$.  As each value of $f_1(x)$ gives
a unique string $x$ and we have considered
all possibilities of inserting a bit, the theorem 
follows.
\end{proof}

We now give a diferent, almost as good construction, using
the run based function, to develop some ideas that will be useful for later.
\begin{theorem}\label{thm:1sumr}
In the case of one deletion
the value of $f_1^r(x)$ modulo $2n+2$ jointly with $y$ determines
$x$ uniquely.
\end{theorem}

\begin{proof}
Let $y$ be the string obtained by deleting a bit from $x_0x_1 \dots x_{n+1}$ which is not the artificial bits $x_0=0$ and $x_{n+1}=1$ at either end. 
When we insert a bit into $y$ we either create a new run or not. 

If the inserted bit is inserted without creating
a new run we have $f_1^r(x)=f_1^r(y)+r$ if and only if the bit
is inserted in the $r$'th run of $y$. 

On the other hand, if the inserted bit creates a run, the smallest increase in $f_1^r(y)$ is obtained by inserting this bit just to the left of $x_{n+1}$ (the last bit of $y$). Note that since $x_{n+1}=1$, for this to happen we must have $y_{n-1}=1$ and the inserted bit must be a \zero. In this case we have $f_1^r(x)=f_1^r(y)+r_s+3$ where $r_s$ is the rank of $x_{n+1}$ in $y$ (the rank of the inserted bit is $r_s+1$ and the rank of $x_{n+1}$ increases from $r_s$ to to $r_s+2$).

Note that even this minimum increase in the run-creating case is strictly larger than the maximum increase possible when we do not create a new run, so there can be no clash in the value of $f_1^r(x)$ between these two cases.

When we move this inserted bit left to the next position
where it can be placed between two equal bits
we get the next possible placement for a run-creating bit.
Suppose we need to move the bit passed $t$ bits (which
must be alternating),
then each of these elements increase their
rank by 2 while the rank of the moving element
decreases by $t-1$.  Thus there is a net increase
of  $f_1^r(x)$ by $t+1$ and in particular there
is a strict increase.

The maximal total increase in $f_1^r(y)$ is achieved when the moving
bit is placed between the first two equal bits.
If this happens between $y_i$ and $y_i+1$ then
$y_i$ has rank $i$, the inserted bit gets rank
$i+1$ and we have $n-i$ bits that have
increased their ranks by 2.  We thus have $f_1^r(x)=f_1^r(y)+i+1+2(n-i)
\leq f_1^r(y)+1+2n$. Thus knowing $y$ and $f_1^r(x) \mod (2n+2)$ suffices to construct the integer value $f_1^r(x)$. As we have considered
all possibilities of inserting a bit each of which gives a different value of $f_1^r(x)$, the theorem 
follows.
\end{proof}

\begin{obs}\label{obs:increase}
When we move
a run-creating bit left, the
value of $r_i$ is non-decreasing for each $i$. In fact, the rank
increases by one for each position passed
by the moving bit (the rank of the bit we pass over actually 
increases by two, but as the bit also moves
one position to the right and the bits are
alternating, the increase is one compared to the element previously
in that position). Therefore the moving bit gets a rank that is
one more than the rank of the element previously in the same
position.
\end{obs}

\gap
\section{Correcting two deletions with size-$2$ lists}
\label{sec:list-decode}
\gap
In this section, we move to the case of two deletions, so $y$ is a subsequence with two bits from $x$ deleted. We will prove that knowing
$f_1^r(x)$ and $f_2^r(x)$ allows us to list decode the subsequence $y$ with list size 2, i.e., pin down $x$ to one of two possible strings. 
We start with a very simple lemma that will spare us some
calculations.  We skip the proof which is a simple consequence of convexity.

\begin{lemma} \label{lemma:convex}
Let $a_i$ and $a_i'$ be two sequences of non-negative integers
such that $\sum a_i = \sum a_i'$ and there is a value $t$
such that for all $i$ such that $a_i < a_i'$ we have
$a_i' \leq t$ and for all $i$ such that $a_i > a_i'$ we have
$a_i' \geq t$.  Then, unless the two sequences are equal,
$\sum a_i(a_i-1) > \sum a_i'(a_i'-1)$.
\end{lemma}

Returning to the main theme of the paper,
we first give some situations where we have unique
decodability.  Note that we insert two bits, we create either zero, two, or four runs.
In two of these cases, it is easy to identify $x$ uniquely.

\begin{lemma}\label{lemma:04run}
Suppose we add zero or four new runs when inserting the two bits.
Then $f^r_1(x)$
and $f^r_2(x)$ jointly with $y$ determine $x$ uniquely.
\end{lemma}

Note that as we assume that the we know the total number of runs,
we can tell when the condition of the lemma is true.

\begin{proof}
Suppose first that neither of the two inserted bits create
a new run.  If the two bits are inserted into runs $r_1$ and $r_2$
respectively, with $r_1 \leq r_2$,
then given the value of $r_1+r_2$ and $r_1(r_1-1)+r_2(r_2-1)$
it is easy to reconstruct $r_1$ and $r_2$.

In the case when four runs are created, place the bits
as close to each other as possible that results in the correct value of $f_1^r(x)$.
As two adjacent bits cannot both create two new runs there
is some separation between the bits.

As the value of $f_1^r$ is strictly increasing when 
run-creating bits
move left, all other possible insertion locations yielding the same value for $f_1^r(x)$
is obtained by moving the leftmost bit left and the rightmost bit
right. It follows from 
Lemma~\ref{lemma:convex} and Observation~\ref{obs:increase} that
the value of $f_2^r(\cdot)$ is strictly decreasing during such moves.
The implies that the configuration obtaining the correct value for
$f_2^r(x)$ is unique and Lemma~\ref{lemma:04run} follows.
\end{proof}

In the case of when only two new runs are created by the re-insertion of the deleted bits, we do sometimes get
some ambiguity. However, we can pin down the string $x$ to one of two possibilities. 
\begin{lemma}\label{lemma:2run}
Suppose the insertion of two bits in $y$ to recover $x$ adds two new runs.
Then given $y$, there are at most two values
of $x$ that can have the same values of $f^r_1(x)$
and $f^r_2(x)$.
\end{lemma}

\begin{proof}
We claim that we can view the process of inserting
the two missing bits as we first add one bit
causing two new runs and then add the second bit
not causing any new run.  It is easy to see that
this is possible unless the two bits are inserted
next to each other and surrounded by two unequal bits.
From now on we use the bold font to indicate inserted
bits and thus the situation we describe is given, for example, by
00{\bf 10}1.  Independently on which order we insert the bits, 
it is the second bit that creates the extra runs.

In this case, however,
we can create the same string as 0{\bf 01}01 or
001{\bf 01}, by only
changing the identity of the inserted bits.
In either of these cases it is possible to first insert
a bit creating two new runs.  This works in general --- whenever the two 
added bits are part of a sequence of alternating bits
we make sure that the two added bits are at the beginning
or end of this sequence.

The process is thus that we start with the string $y$ and
add a bit $b_0$ causing
two new runs.  We now compute the rank of $b_1$ to 
add to get the correct value of $f_1^r(x)$.  If this rank
is smaller than 0 or larger than the maximal rank of the
created string, then the
placement of $b_0$ was impossible. Otherwise, we have
a unique run into which to insert $b_1$.

Let us first place $b_0$ as far right as possible giving
a possible placement of $b_1$.  Let us see what
happens when we move $b_0$ from any position to the 
first position to its left where it also creates two
runs.  We call this an {\em elementary move}.
The bit $b_0$ starts between two equal bits
and after the elementary move, it shifts to the first place to its left where it can again be placed between two
adjacent equal bits.  We note that the value of the bit $b_0$ might change from $0$ to $1$ after the elementary move. Suppose $b_0$ passes $t$
alternating bits in such an elementary move. This decreases its rank by $t-1$
and increases the rank of each of the $t$ passed elements by $2$ each. 
To compensate for this net increase in total rank of $t+1$, $b_1$ must decrease its rank by $t+1$
to maintain the value of $f_1^r(x)$.
To achieve this, $b_1$ must move past at least $t$ bits
from $y$ (it might pass $b_0$) so it always passes at least as many
bits in $y$ as does $b_0$.  Thus if $b_1$ starts to the
left of $b_0$ it remains to the left.  If it starts to the
right, it might overtake $b_0$ once but after this
it remains to the left of $b_0$.

Suppose the string before an elementary move is $x$ and after 
the same move it is $x'$.
We have three cases.

\begin{enumerate}
\vspace{-1ex}
\itemsep=0ex
\item $b_1$ is to the left of $b_0$ in both $x$ and $x'$.

\item $b_1$ is to the right of $b_0$ in both $x$ and $x'$.

\item $b_1$ is to the right of $b_0$ in $x$ but
to the left in  $x'$.

\end{enumerate}
We have the following claim.
\begin{claim}\label{claim:monotone}
In in Case 1, $f_2^r(x) >f_2^r(x')$, 
and in Case 2 $f_2^r(x) <f_2^r(x')$.
\end{claim}

Before we establish this claim let us see that it finishes
the proof of Lemma~\ref{lemma:2run}.  The claim says
that $f_2^r$ is strictly monotone before
and after the take-over.  This implies that each fixed value of $f_2^r(x)$
can only be achieved once before and once after the take-over for
a total of at most two times.
We move to establish Claim~\ref{claim:monotone}.

In the two cases under consideration
$b_1$ does not pass $b_0$.  There can, however, be some positions
passed by both bits but if we move the leftmost bit
first, the bits do not really interact.

Let us first see what happens to the ranks for
old elements.  Clearly the ranks and positions
of elements not passed by either moving element
remain the same and we focus on the more
interesting elements.

\begin{itemize}
\itemsep=0ex
\vspace{-1ex}
\item Elements passed by only $b_0$ move one
position to the right and increase their
rank by two.  Thus if you compare their
ranks to that of the element previously in the same
position it increases by at least $1$.

\item Elements passed by only $b_1$ move one
position to the right and keep the same
rank.  Thus if compared to the element
previously in the same position their
rank remains the same or decreases by $1$.

\item Elements passed by both elements move two
positions to the right and increase their
rank by two. As $b_0$ moves passed alternating
elements their new rank is equal to the
rank of the old element in the same position.

\item $b_0$ gets a rank that is one larger
than the old element in the same position.

\item $b_1$ gets a rank equal to that of 
the old element in the same position, unless
it is in a position passed by $b_0$ in which case
its rank is greater than the rank of the
element previously in the same position.

\end{itemize}
We see that all positions with
decreased ranks are passed only by $b_1$.
From this it follow that, in Case 1,
we have the all positions with decreased
rank are to the left of all positions with
increased ranks.  The claim in this case now
follows from Lemma~\ref{lemma:convex}.

Similarly, in Case 2, all positions with decreased
rank are to the right of all positions with
increased ranks and we conclude that 
the claim is once again true. Appealing to Lemma~\ref{lemma:convex}, this completes the proof of
Lemma~\ref{lemma:2run}.
\end{proof}

We summarize the two lemmas into a theorem. 

\begin{theorem}
	Let $y \in \{0,1\}^{n-2}$. 
	There can be at most two strings $x \in \{0,1\}^n$ that have $y$ as a  subsequence and which share a common value of $f_1^r(x)$, $f_2^r(x)$, and total number of runs. 
\end{theorem}

Using the connection outlined in Lemma~\ref{lem:reduce-to-sketch} between recovery from known sketches $f_1^r(x)$ and $f_2^r(x)$ and deletion-correcting encodings, and the fact that $f_1^r$ can be specified modulo $O(n)$, $f_2^r$ modulo $O(n^2)$, and the number of runs modulo $O(1)$, we have our result on list-decodable codes for two deletions. We remind the reader that the existence of such list-decodable codes of size asymptotically bigger than $2^n/n^4$ was not known prior to our work. 
	\begin{theorem}
		\label{thm:list2}
		There is a 2-deletion code of size $\Omega ( 2^n/n^{3} )$
		that is list-decodable with list size 2.
	\end{theorem}
Since our sketches $f_1^r(x)$ and $f_2^r(x)$ are simple explicit functions, by Lemma~\ref{lem:reduce-to-sketch} we can get explicit codes with $O(\log \log n)$ extra redundant bits.
	\begin{corollary}
		\label{cor:list2}
		There is an explicit (efficient encodable) 2-deletion code of size $\Omega ( 2^n n^{-3} (\log n)^{-O(1)})$
		that is list-decodable with list size 2.
	\end{corollary}
\medskip
\noindent
{\bf Remark.}  Let us give an example to 
show that we do have list size
two in many situations.
Take any numbers $t_0$ and $t$ and consider the following
two ways inserting two bits.

\begin{itemize}
\vspace{-1ex}
\itemsep=0ex
\item Insert a bit not creating
a run in position $t_0-3t$ and one creating a run
in position $t_0-t$.
\item Insert a bit not creating a run
in position $t_0+3t$ and one creating a run in
position $t_0+t$.
\end{itemize}
We make an approximate calculation for the
difference of $f_1^r$ and $f_2^r$ between
these two ways of inserting the bits.
Let us for simplicity
assume that an element in position $i$ has rank exactly $i/2$
and ignore the difference
between $r_i (r_i-1)/2$ and $r_i^2/2$ in the definition
of $f_2^r$.  

As the ranks of all existing elements to the left of
position $t_0-t$ and
to the right of position $t_0+t$ are the same after the insertions,
we ignore them when we calculate the the values of 
$f_1^r(x)$ and $f_2^r(x)$ and we sum up only terms that
are different in the two sums.
\begin{itemize}
\vspace{-1ex}
\itemsep=0ex
\item In the first case the inserted elements
get ranks $(t_0-3t)/2$, and $(t_0-t)/2$, respectively.
All existing elements between positions $t_0-t$ and $t_0+t$
get their ranks increased by 2.  This implies that the
increase in $f_1^r(x)$ is roughly
$$
(t_0-3t)/2+(t_0-t)/2+4t= t_0+2t \ .
$$
As the average rank of the elements increasing their
ranks by two is $t_0/2$ the increase in 
$f_2^r(x)$ is roughly
$$
\frac 12 ((t_0-3t)/2)^2+ \frac 12 ((t_0-t)/2)^2+2t \cdot 2 \cdot t_0/2 =
\frac 14 (t_0^2+4t t_0+5t^2)
$$

\item In the second case the inserted elements
get ranks about  $(t_0+t)/2$ and $(t_0+3t)/2$, respectively
while no existing elements before position $t_0+t$
change their ranks.  This implies that the
increase in $f_1^r(x)$ is roughly
$$
(t_0+t)/2+(t_0+3t)/2= t_0+2t,
$$
and the increase in $f_2^r(x)$ is roughly
$$
\frac 12 ((t_0+t)/2)^2+\frac 12((t_0+3t)/2)^2= 
\frac 14 (t_0^2+4tt_0+5t^2),
$$
both matching the first case.
\end{itemize}

It is easy to construct situations where we
get an exact match.  Looking at the example we see that the
non run-creating bit will pass the run-creating bit around
position $t_0$ and it is natural that at equal times
before and after this event we get about the same value
for $f_2^r$.

We believe that the given family of examples 
giving the same values for $f_1^r(x)$
and $f_2^r(x)$ are essentially all such examples.
We do not see any fundamental objection to the existence
of a third function $f^?$ that would be able to
distinguish all such pairs, but we have been
unable to construct such a function with a small range.

In the next section we achieve unique decodability and this
analysis is heavily based on $f_1$.  We invite the reader to
check that $f_1$ is not sufficient to distinguish the two cases
in the example above in general.  If, in the two situations described,
the two leftmost bits are equal and the two rightmost bits also
are equal but different from the first pair, $f_1$ is also approximately preserved.

\section{Unique decodable codes for two deletions}
\label{sec:unique-decode}
We now return to our main goal, namely the construction of a $2$-deletion code
with sketches of size totaling about $4 \log_2 n$ bits.

We return to studying
$f_1(x)$ and our analysis focuses on all ways of inserting the
two bits to get the correct value of $f_1$.  We start 
with some possible configuration and obtain all other
configurations by moving the bits in a way that preserves $f_1$,  and analyze the impact on $f_2$ and
$f_1^r$ in this process.  In this section, we will not require the function $f_2^r$. 

An inserted \one\
decreases the value of $f_1$ if it moves left and passes a \zero\ and increases
the value if it moves right passing a \zero.  
For an inserted \zero, the two cases are reversed.
Remember that, by the proof of Theorem~\ref{thm:1sum}, once we have
placed one of the bits, the location of the other
bit is uniquely determined.

\gap
\subsection{When two \zero's or two \one's go missing}

We start with the easy case, analogous to Lemma~\ref{lemma:04run},
when the two deleted bits have the same value.


\begin{lemma}
	\label{lem:same-bits-deleted}
If we have deleted two \one's or two \zero's in forming the subsequence $y$ from $x$, then $f_1(x)$ and
$f_2(x)$ together with $y$ identify $x$ uniquely.
\end{lemma}

\begin{proof}
Suppose we insert two \zero's as close to each other
as possible giving the correct value of $f_1(x)$.  Now all
other insertions of the two bits giving the correct
value of $f_1(x)$ are obtained by moving the \zero\ on the left further
left and the \zero\ on the right further right.  Each bit moves past
one bit of the opposite type and we again call such a move an
elementary move.  At each such step two
terms in the sum  (\ref{eq:f1}) defining $f_1(\cdot)$ change.  The left moving
\zero\ causes one term to increase by one and the right one
causes one term to decrease by one.  It follows by Lemma~\ref{lemma:convex}
that $f_2(x)$ is monotonically and strictly decreasing in this process.
This implies that the location of the two bits giving the
correct value for $f_2(x)$ is unique.

The case of inserting two \one's is completely analogous except that the sign
reverses and $f_2(x)$ is monotonically and strictly \emph{increasing}
as the two \one's move apart.
\end{proof}

Note that we once again assume we know the total number of \one's in $x$ (modulo $3$ say), so we can detect that we are in the case of Lemma~\ref{lem:same-bits-deleted}.

It now remains to consider the case when we have to insert a \zero\ and a \one\ in $y$ to recover $x$. 
It is good to remember that once the moving \zero\ has 
passed a \one, if it runs into one or more \zero's, it moves past
these ``effortlessly'' and stops next to the first \one\ it
encounters.  Similarly a \one\ moves effortlessly past a run of \one's.  The existence of these effortless moves
makes the bits move at (slightly) different speeds.  They move,
on the average, two steps to get past the next bit of the
opposite type but there are some random fluctuations.

\gap
\subsection{Elementary moves and overtaking}
Suppose we insert a \zero\ and a \one\ as far right as possible
and we need to move both left in an elementary move. The following easy to check observation will be handy multiple times.

\begin{obs}
\label{obs:overtake}
When we move the inserted \zero\ and \one\ to the left, the value of $f_2(x)$ decreases if the
\zero\ is to the left of the \one\ and increases otherwise.
Thus to obtain several possibilities where one can place the moving bits with the correct
values of $f_2(x)$ the lead must change between \zero\ and \one.
\end{obs}

As the bits
move at the same speed but with different random fluctuations,
if the bits start close to each other the bits can overtake
each other many times.  This is in contrast to the case
when analyzing $f_1^r$ and $f_2^r$ in Lemma~\ref{lemma:2run} where the bit not causing any new runs moved at a strictly greater speed than the bit causing
new runs.  This fact was the key to the proof of Lemma~\ref{lemma:2run}
and the analogous lemma (with $f_1$ and $f_2$ instead of $f_1^r$ and $f_2^r$)  is not true in the current situation.

As the moving bits can overtake each other we need to be
slightly careful when defining an elementary move. 
\begin{definition}[Elementary move]
We first move the leftmost moving bit left past one
bit of opposite value and then past a run of bits
of its own value until it is adjacent to bit of the opposite
value.  We then repeat this procedure with the second moving bit. 
\end{definition}

Note that both the moving \zero\ and moving \one\ have a bit of opposite value of their immediate left before and after an elementary move. 
These two moves together may not change the string
as can be seen from the following example (again with moving bits in bold).
Suppose the current string is $100{\bf 10}$.  Moving the
first bit produces $10{\bf 1} 0{\bf 0}$ and moving
the second bit we get $1{\bf 0} 0{\bf 1} 0$ the same string
as we started with but the identity of the moving bits
have changed. 
It is easy to see that each of moving bits always
move past at least one old bit (so there is a notion of progress in the position of the moving bits even if the string itself doesn't change in an elementary move).

When the moving bits are adjacent to each other, after the left bit moves, the moving bit on the right will have a bit of the same value to its left which it will also jump over during the elementary move (as happens in the above example).
This example also shows the mechanism of overtaking.  When the bits are close
and move in an area of mostly \zero's, the moving \zero\ moves faster.
It is easy to see that for an elementary move to cause one
moving bit to overtake the other, the bits must have started
next to each other.
 

\medskip 
\noindent
{\bf Remark.} If the
two erased bits are far from each other we do get unique decodability.
We claim that the values of $f_1(x)$ and $f_2(x)$ are,
with high probability sufficient to reconstruct $x$ if this
string is random and 
two random bits are deleted.  This follows as two random
bits are likely to be far apart and the fluctuations in
the speeds is small for a random $x$.

\gap
\subsection{When the run count changes by 0 or 4} 

While the values $f_1(x)$ and $f_2(x)$ might guarantee unique decoding in many cases, we are 
interested in a worst case result and thus we now bring
$f_1^r(x)$ into the picture (which incurs an additional $\log_2 n +O(1)$ bits of redundancy since we can specify $f_1^r$ modulo $O(n)$).  It turns out this information
is sufficient for unique decodability in half of the 
remaining cases.

\begin{lemma}\label{lemma:04run01}
Suppose we add 0 or 4 new runs when inserting a \zero\ and a \one. Then  the information $f_1(x)$ and $f_1^r(x)$ jointly
with $y$ is sufficient to identify $x$ uniquely.
\end{lemma}

\begin{proof}
Let us start with the case of no new runs.  Let us insert the
two bits as far right as right as possible (as allowed by $f_1(x)$) and
let us move bits to the left keeping the correct value of $f_1(x)$.
It easy to see that $f_1^r(x)$ is strictly decreasing
as moving bits to the left that do not create runs
makes $f_1^r$ strictly decrease.

The case of four new runs (i.e. both inserted
bits giving two new runs) is similar.  We again insert the bits as far right
as possible based on $f_1(x)$.  This time when moving
both bits to the left, $f_1^r(x)$ is strictly increasing.
\end{proof}

\subsection{When the run count changes by 2}
\label{subsec:run-2-0-and-1}
We need to analyze the final case when we insert a \one\ and a \zero\
and exactly one of the two bits creates two new runs. This takes the bulk of the work given that Lemmas \ref{lem:same-bits-deleted} and \ref{lemma:04run01} had short and easy proofs.

Since we have to get the value of $f_1(x)$ right, the insertion of a \one\ in position $i$ determines the position at which the \zero\ must be inserted.


For the analysis, we track certain \emph{pseudorank} functions with the property if the \one\ and \zero\ are inserted into positions (with the correct value of $f_1(x)$) where one of them creates two runs and other creates no runs, the pseudorank equals $f_1^r(x)$. We can then track how the pseudorank changes as we move the inserted bits to understand the positions where the correct value of $f_1^r(x)$ can also be obtained.


\begin{definition}[Pseudorank]
\label{def:pseudorank}
The \emph{1-pseudorank}, denoted $A_1(i)$ indexed by the position $i$ where the moving \one\ is inserted (into $y$), is defined by the following process:

\begin{enumerate}
\itemsep=0ex
\vspace{-1ex}
\item Insert a \one\ in position $i$.

\item Insert a \zero\ in a position to ensure that $f_1(x)$
takes the correct value.

\item For each bit of $y$, its 1-pseudorank is its rank
in $y$ unless it is to the right of the inserted \one\ in which
case its 1-pseudorank is two more than its rank in $y$.

\item The 1-pseudoranks of the inserted \one\ and \zero\
equal their actual ranks.

\item Finally, $A_1(i)$ is defined to be the sum of these 1-pseudoranks of the individual bits.
\end{enumerate}
The \emph{0-pseudorank} function $A_0(\cdot)$ is defined analogously, reversing the roles of \one\ and \zero\ (so bits of $y$ to the right of where the \zero\ is inserted have pseudoranks equal to their rank in $y$ plus $2$).
 However, we index $A_0$ also by the position of the inserted \one\ (rather than the inserted \zero) enabling us to reason about and compare $A_0$ and $A_1$ at the same location where we insert the \one.
 
\end{definition}


Note that whenever the 
described process makes the inserted \one\ create two
new runs while the inserted \zero\ does not, $A_1(i)$ agrees
with $f_1^r(x)$.  A similar claim holds for $A_0(i)$ when the inserted \zero\ creates two runs and the inserted \one\ creates no runs. 
The following lemma establishes a crucial monotonicity of the pseudorank functions.

\begin{lemma}\label{lemma:weakincrease}
$A_1(i)$ never decreases by an elementary move.
Whenever at least one of the moving bits encounters a run of at least
two adjacent \one's, $A_1(i)$ strictly increases.  Also if the
 moving \one\ overtakes the moving
\zero\ then $A_1(i)$ strictly increases.

Similarly, $A_0(i)$ never decreases by an elementary move.
Whenever at least one of the moving bits encounters a run of at least two
adjacent \zero's, $A_0(i)$
strictly increases.  Also if the moving \zero\ overtakes the moving
\one\ then $A_0(i)$ strictly increases.
\end{lemma}
\begin{proof}
As \zero's and \one's are symmetric it is enough to prove
the first part of the lemma.  We move the leftmost
bit first and let us first assume that the second moving bit
does not overtake the first in which case we can analyze the
two moving bits independently.

The moving \one\ either moves past a single \zero\ encountering
another \zero, or it moves past a \zero\
and then passes at least one additional \one\ until it hits
the next \zero.  In the first case, the single
\zero\ increases its 1-pseudorank by 2 and no
other rank (including that of the moving \one) changes.  In the second case, the
rank of the moving bit decreases by 2 but there
are at least two bits whose 1-pseudorank
increases by 2.  In either case the total 1-pseudorank
increase of all bits around the moving \one\ (including
itself) is at least 2. Also, if the moving \one\ moves effortlessly past at least two 1's, then the 1-pseudorank increases by at least $4$.

The 1-pseudorank of the moving \zero, which equals its actual rank, stays the same if there is a run of at least two $1$'s immediately to its left of the moving \zero, and it drops by $2$ if there is a run with a single \one\ to its left. 
This is illustrated respectively by the two cases where 11\textbf{0}000
changes to 1\textbf{0}1000 (the 1-pseudorank stays the same) and where
101\textbf{0}00 changes to 1\textbf{0}0100 (the 1-pseudorank decreases by 2).  No other
change of 1-pseudorank is caused by the moving \zero.

Adding together the two
contributions caused by the moving \zero\ and the moving \one,  we see that $A_1$ cannot
decrease.  Further, if at least one of the moving bits encounters a run of two or more $1$'s during the elementary move, we get a strict increase in $A_1$.



We now consider the situation when the second moving bit overtakes
the first.  There are two cases based on which bit is to the left. Suppose that the moving \one\ is to the left and moves first.  It moves exactly
one step to its left (otherwise the moving \zero\ cannot overtake it).  Its rank remains the same and it passes exactly
one bit of $y$ whose 1-pseudorank increased by two.  After this the moving \zero\ decreases its rank
by two and does not change any other 1-pseudorank.  In this case $A_1(i)$
does not change. This case is captured by the transformation $w 0^b \textbf{1} \textbf{0} z \to w \textbf{0} 0^{b-1} \textbf{1} 0 z$ for some $b \ge 2$. Note that neither moving bit encounters a run of two or more $1$'s in such an elementary move.

If the moving \zero\ moves first, it takes one step to its left and does not 
does not change
its rank.  After this the moving \one\ moves past the \one\
just passed by the moving \zero\, passes the moving \zero\ and
additionally at least one \one.   The moving \one\ decreases
its rank (which equals its 1-pseudorank) by two but at least two bits increase their 1-pseudoranks 
by two.  Thus in this case there is a strict increase in
$A_1$.
\end{proof}

The above lemma implies that for $A_1$ (resp. $A_0$) to not increase, both the moving bits must be moving in a region without $11$ (resp. $00$) as a sub-string. 
In view of this, the following definition is natural.  Below $d$ be an absolute constant to
be fixed later (but the choice $d=7$ will suffice).

\begin{definition}[Regularity]
\label{def:regular}
We say that a string $x \in \{0,1\}^n$ is {\em regular} if each (contiguous) sub-string of $x$ of length at least $d\log_2 n$ contains both 00 and 11. 
\end{definition}

In other words, there is no sub-string of length at least $d\log_2 n$ all of whose $1$-runs are of length one, or all of whose $0$-runs are of length one.
As we establish in Lemmas \ref{lem:fibonacci} and \ref{lem:regenc}, when $d$ is a large enough absolute constant, most of the $n$-bit strings are regular and further one can efficiently encode into a large subset of regular strings. So we can focus on deletion recovery assuming that $x$ is regular.

\smallskip Denote by $P_1$ (resp. $P_0$) the set of positions $i$ for which $A_1(i)$ (resp. $A_0(i)$) equals the desired value of $f_1^r(x)$.
By Lemma~\ref{lemma:weakincrease}, we have that $P_1$ is contained in an interval, say $I_1$, that does not contain
two adjacent $1$'s. Similarly, $P_0$ is contained in an interval, say $I_0$, that does not contain $00$. If we are guaranteed that $x$ is regular, then the length of $I_0,I_1$ is at most $d\log n$.

Thus, to get the value of $f_1^r(x)$ correct, the possible locations where the \one\ can be inserted are contained in $I_0 \cup I_1$. We now prove that if we also have to get the correct value of $f_2(x)$, then the positions where $1$ can be inserted must be contained in one of these two intervals $I_0$ or $I_1$.

\begin{lemma}
\label{lem:only-one-interval}
There cannot be $p_1 \in I_1 \setminus I_0$ and $p_0 \in I_0 \setminus I_1$ such that inserting the \one\ at $p_1$ or $p_0$ (and the \zero\ at the corresponding position implied by $f_1(x)$) both leads to the correct values of $f_2(x)$ and $f_1^r(x)$.
\end{lemma}
\begin{proof}
Suppose, for contradiction, there exist positions $p_1 \in I_1 \setminus I_0$ and $p_0 \in I_0 \setminus I_1$ where the \one\ can be inserted both of which lead to correct values of $f_2(x)$ and $f_1^r(x)$.
We have 
\begin{equation}
\label{eq:a0-a1-fr}
A_0(p_0)=A_1(p_1)=f_1^r(x)  \ .
\end{equation} 
Suppose,
without loss of generality, that $p_0$ is to the right
of $p_1$. First note that since $p_0$ is strictly to right
of $I_1$ and $A_1$ is monotone, we have that 
$A_1(p_0) \le A_1(p_1) = f_1^r(x)$. But since  $p_0 \notin I_1$, $A_1(p_0) \neq f_1^r(x)$ and thus $A_1(p_0) < A_1(p_1)=A_0(p_0)$.
Similarly, as $p_1$ is strictly to the left of $I_0$
we have $A_0(p_1) > A_0(p_0)=A_1(p_1)$.

We claim that at any position
such that $A_0 (p) > A_1(p)$, the inserted \one\ is at least
two positions to the right of the inserted \zero.  This follows
from the definition as a \zero-pseudorank
is larger than the \one-pseudorank
only for elements that are to the right of the
inserted \zero\ but to the left of the inserted \one.  Thus
if $A_0(p) > A_1(p)$ we have at least one such element and
the claim follows.  We conclude that the inserted \one\ is
to the right of the inserted \zero\ both at $p_0$ and $p_1$.

Recall Observation~\ref{obs:overtake} that for $f_2(x)$ to return
to a previous value while $f_1(x)$ is preserved,
the lead must change in the
race between the moving \one\ and the moving \zero.
We conclude that this must happen between the positions $p_0$ and
$p_1$.  In fact since the moving \zero\ is to the left
of the moving \one\ at both $p_0$ and $p_1$
there must be at least two such overtaking events
and there must be an elementary move at which the
moving \one\ overtakes the moving \zero.  Call this position $p^\ast$.

We first observe that $A_1(p^\ast)= A_0(p^\ast)$.
This is obvious by definition as each bit has the same $0$-pseudorank and $1$-pseudorank when
the inserted \zero\ and \one\ are next to each other.

By Lemma~\ref{lemma:weakincrease}
whenever the moving \one\ overtakes the moving \zero\
we have a strict increase in $A_1$.  Combining this with
the monotonicity, we have
$$
A_1(p_1) > A_1(p^\ast)=A_0(p^\ast) \geq A_0(p_0)
$$
which contradicts (\ref{eq:a0-a1-fr}). Thus such $p_1, p_0$ cannot exist.
\end{proof}

Therefore, we conclude
that all possible alternatives for inserting the \one\ must fall within either $I_0$
or $I_1$. A similar claim shows that the possible positions to insert the \zero\ must be confined to a single interval that does not contain two adjacent $0$'s. If $x$ is regular, we can conclude that the positions where the \one\ may be inserted is confined to an interval, say $I$, of width $d \log n$, and similarly the possibilities to insert the \zero\ are confined to a width $d \log n$ interval $J$. By Observation~\ref{obs:overtake}, these intervals $I$ and $J$ must intersect so that the lead can change between the moving bits. Thus, both the insertions must be confined to the interval $I \cup J$, which has width $2d \log n$.\footnote{We can in fact claim that both insertions happen in an interval of size $d \log n$, namely either $I$ or $J$, but this factor $2$ savings is inconsequential.}
We have thus established the following lemma.
\begin{lemma}\label{lemma:bothclose}
Assume that $x$ is regular in the sense of Definition~\ref{def:regular}.
Suppose we add two exactly new runs when inserting a \zero\ and a \one\ into $y$ to obtain $x$.
Then given $f_1(x)$, $f_1^r(x)$, $f_2(x)$, and $y$
either
\begin{itemize}
\itemsep=0ex
\item The information is sufficient to identify $x$ uniquely, or 

\item There is an interval $I$ of length at most $2d\log n$
such that both insertions in $y$ are located in this
interval.

\end{itemize}

\end{lemma}

\noindent
{\bf Remark.} (Insufficiency of $f_1(x)$, $f_1^r(x)$ and $f_2(x)$ to uniquely pin down $x$.) Let us given an example showing that the information
$f_1(x)$, $f_1^r(x)$ and $f_2(x)$ is not sufficient 
to uniquely determine $x$.  The strings
1101111{\bf 0}10{\bf 1}1 and 11{\bf 1}01{\bf 0}111101 which both can become
1101111101 have the same values for these three functions.
Both insertions are located in the interval $I_0$ which has no two adjacent $0$'s in this case. There are two positions $i,j \in I_0$, $i > j$, with $A_0(i) = A_0(j)=f_1^r(x)$ and the moving \one\ overtakes the moving \zero\ between positions $i$ and $j$.

\subsection{Handling ambiguity within small intervals}

Combining Lemmas~\ref{lem:same-bits-deleted}, \ref{lemma:04run01}, and \ref{lemma:bothclose}, the original string $x$ is either uniquely determined, or we have found an interval of size $\Delta \le O(\log n)$ such that both deletions happened in that interval, provided $x$ is regular.

To recover from this last case, we can encode each of those intervals by a two-deletion code. Since this code is for short lengths, it can either sketches of length $\approx 4 \log_2 \Delta \le O(\log \log n)$ matching the existential bound that is found by brute-force (cf. \cite[Lemma 1]{BGZ-soda}), or an explicit sub-optimal sketch of length  $c \log \Delta \le O(\log \log n)$ for a larger constant $c$, for instance the construction with $c = 7$ from \cite{SRB20}. 
Of course we cannot include this sketch for each interval as that would make the overall sketch way too long, but since we know the deletions are confined to one of the intervals, we can simply XOR all these sketches. There is one small catch in that the two deletions might occur in two adjacent intervals if we pick fixed interval boundaries. This is easily handled by computing the sketches also for another set of intervals which straddle the first set of intervals.

Below we execute this idea by introducing explicit sketches for these intervals based on the ranks of the elements, using also the fact that we only need this in the case of Section~\ref{subsec:run-2-0-and-1}, so we have a self-contained solution that also fits the mold of our other sketches. 
Suppose without loss
of generality that we know that the two bits are inserted
in the interval $I_1$ and we have at least two alternatives.
As the moving \zero\ moves at least as fast as the moving \one\ in
$I_1$ there is single take-over point in the interval.  Similar to the
proof of Lemma~\ref{lemma:2run} we can conclude that $f_2(x)$
is monotone before and after this take-over and thus there
are exactly two alternatives where to insert the moving \one.
Suppose the corresponding strings are $x$ and $x'$ where $x$ has the
moving bits further to the right.

When moving the bits left going from $x$ to $x'$
the moving \zero\ has overtaken
the moving \one\ and the moving \one\ causes two new runs in both
positions  while the moving \zero\ does not cause a new run.  Thus if
we look at the corresponding rank string $r$ and $r'$,
we have:

\begin{itemize}

\item $r_i=r_i'$ to the left of the moving \zero\ in $x'$, not including
this last bit.

\item $r_i\geq r_i'$ between the moving \zero\ of $x'$ (inclusive)
and the moving \one\ of $x'$, not including this last bit.
Let us call this interval $J_1$

\item $r_i \leq r_i'$ between the moving \one\ of $x'$ (non-inclusive)
and the moving \one\ of $x$, including this last bit.
Let us call this interval $J_2$

\item $r_i\geq r_i'$ between the moving \one\ of $x$ (non-inclusive)
and the moving \zero\ of $x$, including this last bit.
Let us call this interval $J_3$

\item $r_i=r_i'$ to the right of the moving \zero\ in $x$.

\end{itemize}

Looking at the values $r_i$ and $r_i'$, we have integers $a$,
$b$, $c$ and $d$
such that $r_i, r_i' \in [a,b]$ when $i \in J_1$,
$r_i, r_i' \in [b,c]$ when $i \in J_2$,  and
$r_i, r_i' \in [c,d]$ when $i \in J_3$.  Now suppose
have a function $P$ such $P(i) > P(i+1)$ when $i \in [a, b-1]$
or $i \in [c, d-1]$ while $P(i) < P(i+1)$ when $i \in [b, c-1]$.
The by the above reasoning we have
\begin{equation} \label{eq:diff}
\sum_{i=1}^{n+1}
 P (r_i) < \sum_{i=1}^{n+1} P(r_i').
\end{equation}
This follows as whenever $r_i \not= r_i'$ we have
$P(r_i) < P (r_i')$ by the properties of $P$.

Now we claim that it is possible to find a polynomial
of degree three with the required properties.
Indeed we can take the
quadratic polynomial $Q(x)=(x-b)(c-x)$ which is positive
exactly in the interval $[b,c]$ and demand that the
derivative of $P$ equals $Q(x)$.

We conclude that if we define $f_3^r = \sum_{i=1}^{n+1} r_i(r_i-1)(r_i-2)/6$
then we must have $f_j^r(x) \not= f_j^r(x')$ for at least one $j\in \{1,2,3\}$. Indeed, otherwise the sum of any cubic polynomial
in $r_i$ would be the same at $x$ and $x'$ contradicting
(\ref{eq:diff}).
Thus if we specify these three numbers, then we can distinguish
the two remaining cases.  This seems very expensive but it is
sufficient to specify these numbers locally as follows.

\begin{itemize}

\item
Divide $x$ into blocks of length $2d\log n$, $x^1$,
$x^2, \ldots, x^m$ where $m= \lceil n/ 2d \log n \rceil$ and
we pad the last block with zeroes to make it full length.  

\item Output  $F_2^r(x)=\oplus_{i=1}^m f_2^r(x^{(i)})$, and
$F_3^r(x) =\oplus_{i=1}^m f_3^r(x^{(i)})$ where $\oplus$ is
bitwise exclusive-or.
\end{itemize}

If the interval, $I$, where the two bits are to be inserted
as specified by Lemma~\ref{lemma:bothclose}, fall completely within one
block $x^i$ then $F_2^r(x)$ and $F_3^r(x)$
makes it possible to reconstruct $x$ uniquely.
This follows as if both bits are inserted in $x^{(i)}$ then
we know $x^{(j)}$ for $j\not =i$ and we can compute $f_2^r(x^{(j)})$
and $f_3^r(x^{(j)})$ and
hence deduce $f_2^3(x^{(i)})$ and $f_3^r(x^{(i)})$.   It is not difficult to
see that we can deduce $f_1^r(x^{(i)})$
from $f_1^r(x)$.  As discussed above
this information makes it possible to distinguish the two
alternatives for $x$ to give a unique reconstruction.

This does not work if the interval where to insert the two bits
intersects two blocks.  We remedy this by making a different division
into blocks of size $2d \log n$, but shifted $d \log n$ 
positions and compute the quantity similar to 
$F_2^r(x)$ and $F_3^r(x)$ with this block division.  The
interval of uncertainty is fully contained in a single block
in one of the two block divisions.

As the lengths of the blocks are $O(\log n)$ it is not difficult
to see that it is sufficient to specify $O( (\log n)^2)$ different values for 
$f_2^r(x^i)$ and
$O( (\log n)^3)$ different values for $f_3^r(x^{(i)})$ to identify the correct
value over the integers.  This gives a total of 
$O( (\log n)^5)$ different values for each block division
and we have finally proved the following theorem. The referred to sketch includes $f_1(x)$, $f_2(x)$, $f_1^r(x)$, and the local sketches $F_2^r(x)$ and $F_3^r(x)$ for the two divisions of the positions into intervals of size $2d \log n$, and any constant sized sketches to determine which case we fall in (regarding identity of the bits deleted and their effect on the number of runs).

\begin{theorem}
\label{thm:main-0}
There is an explicitly computable sketch function $s$ mapping $n$ bits to 
$4 \log n + 10 \log \log n + O(1)$ bits such that for any \emph{regular} $x \in \{0,1\}^n$, given $s(x)$ and any subsequence $y$ of $x$ obtained by deleting two bits, one can uniquely recover $x$.
\end{theorem}

Note that the above only works for regular strings. Appealing to Lemma~\ref{lem:reduce-to-sketch}, we will have our desired two-deletion code if we can encode messages into regular strings. We show how to do this in Section~\ref{subsec:encode-regular}, leading finally to our main theorem giving explicit two-deletions codes of size matching the best known existential bound up to lower order terms.

\begin{theorem}[Main]
\label{thm:main-body}
There is an explicit (efficiently encodable) binary code $C \subseteq \{0,1\}^n$ of size $\Omega (2^n n^{-4}  (\log n)^{-10})$ that can be uniquely decoded from two deletions. 
\end{theorem}

\subsection{Encoding into regular strings}
\label{subsec:encode-regular}

All that remains to be done to complete the proof of Theorem~\ref{thm:main-body} is a way to efficiently encode into regular strings in $\{0,1\}^n$ in a rate-efficient manner. 

First let us show that the number of regular strings is large.
The following lemma shows that non-regular strings form a negligible (exponentially small) fraction of all strings since the $m$'th Fibonacci number $F_m$ is at most $(1.62)^m$.

\begin{lemma}
\label{lem:fibonacci}
The number of $m$ bit strings not containing two adjacent
\zero's is $F_{m+2}$ where $F_i$ is the $i$'th Fibonacci number.
\end{lemma}
\begin{proof}
Let $S_m$ be the number of strings of the type
described in the lemma.  It is immediate
to check that $S_1=2$ and $S_2=3$.  For general $m$ a string of
$m$ bits without 00 is either a string starting with 1 and
an arbitrary string with the property of length $m-1$ or 
a string starting with 01 follow by such a string of length
$m-2$.  We conclude that $S_m=S_{m-1}+S_{m-2}$ and the
lemma follows.
\end{proof}

We can now encode into a large subset of regular strings by enumerating strings of length $O(\log n)$ that contain both $00$ and $11$ and using these locally to piece together an $n$-bit string.

\begin{lemma}
\label{lem:regenc}
There is a one-to-one map $\text{RegEnc}: \{1,2,\dots,M\} \to \{0,1\}^n$ for $M \ge 2^{n-1}$ such that that $\text{RegEnc}$ is computable in $\text{poly}(n)$ time and its image is contained in the set of regular strings (per Definition~\ref{def:regular} with the choice $d=7$).
\end{lemma}
\begin{proof}
Let $Q$ be the set of binary strings of length $\Delta := \lfloor \frac{d}{2} \log_2 n \rfloor$ which contain both $00$ and $11$ as a substring. Denote $m =\lfloor n/\Delta \rfloor$, and $M  = |Q|^m 2^{n-m\Delta}$. 
By Lemma~\ref{lem:fibonacci}, we have 
\[ M \ge 2^n \Bigl(1 -  m (1.62)^2 (1.62/2)^{\Delta} \Bigr) \ge 2^n (1 - 3 \Delta^{-1} n^{1-0.15 d}) \ge 2^{n-1} \]
for $d \ge 7$ and $n$ big enough. Fix any efficiently computable bijection $\phi$ from $[M] := \{1,2,\dots,M
\}$ to $Q^m \times \{0,1\}^{n-m\Delta}$. Consider the map $\psi : Q^m \times \{0,1\}^{n-m\Delta} \to \{0,1\}^n$ that enumerates the strings in $Q$ corresponding to the first $m$ components and then concatenates it with the last $n-m\Delta$ bits to form an $n$-bit string. The composition of $\psi \circ \phi$ is our desired map $\text{RegEnc}$. The regularity of the output string follows because any contiguous substring of length $d \log n$ must include a string from $Q$ and thus have both a $00$ and a $11$ occurring within it.
\end{proof}

\bibliographystyle{alpha}
\bibliography{two-deletions}

\end{document}